\documentclass[submission,copyright,creativecommons]{eptcs}
\providecommand{\event}{TARK 2017} % Name of the event you are submitting to
\usepackage{booktabs} % For formal tables
\usepackage{graphicx}
\usepackage{latexsym}
\usepackage{longtable}
\usepackage{multirow,booktabs}
\usepackage{amsfonts,amsmath,amssymb}
\usepackage{longtable}
\usepackage{multirow,booktabs}
\usepackage{amsthm,amsfonts,amsmath,amssymb}
\usepackage{url}
\usepackage{hyperref}
\usepackage[utf8]{inputenc}
\usepackage[normalem]{ulem}
\usepackage{tikz}
\usepackage{comment}

\newtheorem{definition}{Definition}
\newtheorem{example}{Example}
\newtheorem{proposition}{Proposition}
\newtheorem{fact}{Fact}
\newtheorem{lemma}{Lemma}
\newtheorem{theorem}{Theorem}

\renewcommand{\L}{\mathcal{L}}
\newcommand{\N}{N}

\newcommand{\I}{\Atoms}

\renewcommand{\S}{\mathcal{A}}

\newcommand{\Atoms}{{\bf P}}
\newcommand{\C}{{\bf C}}
\renewcommand{\O}{{\bf O}}

\newcommand{\AB}{\mathcal{O}^\ast}
\newcommand{\maj}{\mathsf{maj}}
\newcommand{\pv}{\mathsf{pv}}

\newcommand{\tuple}[1]{\left\langle #1 \right\rangle}
\newcommand{\set}[1]{\left\{ #1 \right\}}

\newcommand{\0}{{\bf 0}}
\newcommand{\1}{{\bf 1}}

\renewcommand{\phi}{\varphi}

\renewcommand{\phi}{\varphi}

\newcommand{\limp}{\rightarrow}
\newcommand{\AND}{~\mathit{and}~}
\newcommand{\IFF}{\Longleftrightarrow}

\title{Binary Voting with Delegable Proxy: \\ An Analysis of Liquid Democracy
\thanks{The paper outlines work presented at: the {\em Dynamics in Logic IV} workshop, TU Delft, November 2016; the seminars of the Computer Science Departments of the University of Leicester and the University of Oxford, December 2016; the {\em Dutch Social Choice Colloquium}, December 2016. The authors wish to thank the participants of the above workshops and seminars for many helpful suggestions. This paper supersedes the earlier technical report \cite{christoff16liquid}. The authors wish also to thank Umberto Grandi for many insightful comments on an earlier version of this paper.
Both authors acknowledge support for this research by EPSRC under grant EP/M015815/1. 
Zo\'e Christoff also acknowledges support from the Deutsche Forschungsgemeinschaft (DFG) and
Grantov\'{a} agentura \v{C}esk\'{e} republiky (GA\v{C}R) joint project 
RO 4548/6--1.
} 
}

\author{Zo\'e Christoff
\institute{Department of Philosophy \\ University of Bayreuth, Germany}
\email{zoe.christoff@uni-bayreuth.de}
\and
Davide Grossi
\institute{Department of Computer Science \\ University of Liverpool, UK}
\email{d.grossi@liverpool.ac.uk}
}
\def\titlerunning{Binary Voting with Delegable Proxy: An Analysis of Liquid Democracy}
\def\authorrunning{Z. Christoff \& D. Grossi}
\begin{document}

\maketitle

%%%%%%%%%%%%%%%%%%%%%%%%%%%%%%%%%%

\begin{abstract}
%Using tools from binary aggregation, 
The paper provides an analysis of the voting method known as delegable proxy voting, or liquid democracy. The analysis first positions liquid democracy within the theory of binary aggregation. It then focuses on two issues of the system: the occurrence of delegation cycles; and the effect of delegations on individual rationality when voting on logically interdependent propositions.
%identifies some strengths and weaknesses of the method, and  
%sketches proposals
It finally points to proposals on how the system may be modified in order to address the above issues.
%to make it, at least in principle, more resilient towards the occurrence of delegation cycles, as well as better-%behaved when applied for voting on logically interdependent issues.
%in preserving integrity constraints.
%towards the representation of individually irrational opinions.
\end{abstract}

%%%%%%%%%%%%%%%%%%%%%%%%%%%%%%%%%%%

\section{Introduction}

Liquid  democracy \cite{liquid_feedback} is a form of group decision-making considered to lie between direct and representative democracy. It has been used and popularized by campaigns for democratic reforms (e.g., Make Your Laws\footnote{\url{www.makeyourlaws.org}} in the US) and parties (e.g., Demoex\footnote{\url{demoex.se/en/}} in Sweden, and Piratenpartei\footnote{\url{www.piratenpartei.de}} in Germany), which used it to coordinate the behavior of party representatives in local as well as national assemblies. At its heart is voting via a delegable proxy, also called transferable or transitive proxy. For each issue submitted to vote, each agent can either cast its own vote, or it can delegate its vote to another agent---a proxy---and that agent can delegate in turn to yet another agent, and so on. This differentiates liquid democracy from standard proxy voting \cite{Miller_1969,Tullock_1992}, where proxies cannot delegate their vote further. Finally, the agents that decided not to delegate their votes cast their ballots (e.g., under majority rule), but their votes now carry a weight consisting of the number of all agents that, directly or indirectly, entrusted them with their vote.

\paragraph{Context}
Voting by delegable proxy was most probably first outlined in \cite{dodgson84principles}.
Analyses of standard (non-delegable) proxy voting from a social choice-theoretic perspective---specifically through the theory of spatial voting---have been put forth in \cite{Alger_2006} and \cite{Green_Armytage_2014}. 
To date, little work has focused directly on liquid democracy:
\cite{kling15voting} provided an empirical study of voting behavior in liquid democracy based on election data from the Liquid Feedback\footnote{\url{www.liquidfeedback.org}} platform of the German Piratenpartei;
and \cite{skowron16proportional} studied how, in the Liquid Feedback platform, issues to be submitted to vote are selected among user-generated proposals via proportional rankings.\footnote{Another, somewhat tangential work is \cite{Boldi_2011}, which focused on algorithmic aspects of a variant of liquid democracy, called {\em viscous democracy}, with applications to recommender systems.}
However, to the best of our knowledge, no work has so far studied voting by delegable proxy as an aggregation rule in its own sake. We do this in the present paper, studying liquid democracy from the perspective of binary aggregation \cite{Dokow_2010,grandi13lifting,Grossi_2014,endriss16judgment}.

%In proxy voting agents may decide to delegate their vote to exactly one other agent, or to %exercise their voting right themselves. Delegations determine therefore a graph, and it should %be clear that such graph has the same properties of the influence graphs studied in the earlier %sections (they are serial and functional): every agent delegates to exactly one other agent %(possibly itself) who becomes the agent's trustee. So in proxy voting voters are effectively %only those agents that entrusted themselves with the vote, and their vote carries as weight the %cardinality of the set of all agents connected to them by a path in the delegation graph.

%\paragraph{Contribution}
%The objective of the paper is to provide a first analysis, via formal methods, of the liquid democracy voting %system based on delegable proxy. This, we hope, should point to a number of future lines of research and %stimulate further investigations into this and related systems.

\paragraph{Outline}
The paper starts in Section \ref{sec:preliminaries} by introducing some preliminaries on the theory of binary aggregation. This preliminary section presents also novel results on binary aggregation with abstentions, which are needed for the analysis developed later in the paper. Section \ref{sec:proxy} introduces a simple model of liquid democracy based on binary aggregation. Section \ref{sec:extensions} establishes formal relations between the proposed model of liquid democracy and standard binary aggregation with abstentions. It studies the issue of circular delegations, and the issue of individual (ir)rationality when voting takes place on logically interdependent issues. The section finally moves from the analysis provided
%studies voting in liquid democracy from the point of view of the delegation of voting power.  This allows us to %shed light on some issues involved in the liquid democracy system, in particular: 
%Section \ref{sec:extensions} moves from the highlighted limitations 
to outline two variants of delegable proxy, which: are more resilient against delegation cycles (Section \ref{sec:default}); better preserve individual rationality when voting on logically interdependent issues (Section \ref{sec:diffusion}). Section \ref{sec:conclusions} concludes.

%%%%%%%%%%%%%%%%%%%%%%% preliminaries

\section{Binary Aggregation} \label{sec:preliminaries}

The formalism of choice for the analysis presented in this paper is binary aggregation with abstentions (see, for instance, \cite{Dokow_2010}).
%\footnote{The standard framework of binary aggregation without abstention is sketched in the appendix for ease of reference.} 
This section is devoted to its introduction.

\subsection{Opinions and Opinion Profiles}

A binary aggregation structure (\emph{BA structure}) is a tuple $\S = \tuple{\N,\Atoms,\gamma}$ where:
\begin{itemize}
%[noitemsep]
\item $\N = \set{1,\dots,n}$ is a non-empty finite set individuals ($|\N|= n$);
\item $\Atoms = \set{p_1,\dots,p_m}$ is a non-empty finite set of issues or propositions ($|\Atoms|= m$);
%, each represented by a propositional atom;
\item $\gamma \in \L$ is an (integrity) constraint, where $\L$ is the propositional language constructed by closing $\Atoms$ under a functionally complete set of Boolean connectives (e.g., $\set{\neg, \wedge}$).
\end{itemize}
A binary {\em opinion} is an assignment of acceptance/rejection values (or, truth values) to the set of issues $\Atoms$.  Allowing abstention amounts to considering incomplete opinions: an {\em incomplete opinion} is a partial function from $\Atoms$ to $\set{\0,\1}$. We will study it as a function $O: \Atoms \rightarrow \set{\0,\1, \ast}$ thereby explicitly denoting the undetermined value ``$\ast$" corresponding to abstention. Thus, $O(p)=\0$ (respectively, \mbox{$O(p)=\1$}) indicates that opinion $O$ rejects (respectively, accepts) the issue $p$. Syntactically, the two opinions correspond to the truth of the literals $p$ or $\neg p$. For $p \in \Atoms$ we write $\pm p$ to denote one element from $\set{p, \neg p}$, and $\pm \Atoms$ to denote $\bigcup_{p\in\Atoms} \set{p, \neg p}$, which we will refer to as the {\em agenda} of $\S$.

We say that the incomplete opinion of an agent $i$ is \emph{consistent} if the set of formulas $\set{p \mid O_i(p) = \1} \cup \set{\neg p \mid O_i(p) = \0} \cup \set{\gamma}$ can be extended to a model of $\gamma$ (in other words, if the set is satisfiable). Intuitively, the consistency of an incomplete opinion means that the integrity constraint is consistent with $i$'s opinion on the issues she does not abstain about. We also say that an incomplete opinion is {\em closed} whenever the following is the case: {\em if} the set of propositional formulas $\set{p \mid O_i(p) = \1} \cup \set{\neg p \mid O_i(p) = \0} \cup \set{\gamma}$ logically implies $p$ (respectively, $\neg p$), {\em then} $O_i(p) = \1$ (respectively, $O_i(p) = \0$). That is, individual opinions are closed under logical consequence or, in other words, agents cannot abstain on issues whose acceptance or rejection is dictated by their expressed opinions on other issues. The set of incomplete opinions is denoted $\AB$ and the set of consistent and closed incomplete opinions $\AB_c$. 
%As the latter are the opinions we are interested in, 
We will often refer to the latter simply as individual opinions, as they are the ones we focus on.

An \emph{opinion profile} $\O = (O_1,\dots,O_{n})$ records the opinion on the elements of $\Atoms$, of every individual in $\N$. Given a profile $\O$ the $i^{\mathit{th}}$ projection $\O$ is denoted $O_i$ (i.e., the opinion of agent $i$ in profile $\O$). 
%Let us introduce some more notation.
We also denote by $\O(p)= \set{i \in \N \mid O_{i}(p)= \1}$ the set of agents accepting issue $p$ in profile $\O$, by $\O(\neg p)= \set{i \in \N \mid O_{i}(p)= \0}$ the set of agents rejecting $p$ in $\O$, and by $\O(\pm p) = \O(p) \cup \O(\neg p)$ the set of non-abstaining agents in $\O$. Sometimes we restrict the previous definitions to a coalition $C \subseteq \N$, so that $\O_C(p)$ (resp., $\O_C(\neg p)$) denotes the set of agents in $C$ that accept (resp., reject) $p$. Finally, we write $\O =_{-i} \O'$ to denote that the two profiles $\O$ and $\O'$ are identical except, possibly, for the opinion of voter $i$.

\subsection{Aggregators}

An aggregator is a function $F: (\AB_{c})^\N \to \AB$, from profiles of closed and consistent incomplete opinions to incomplete opinions. The \emph{issue-by-issue strict majority rule} ($\maj$) accepts an issue if and only if the majority of the non-abstaining voters accept that issue: 
\begin{align}\label{eq:majast}
\maj(\O)(p)=  
\begin{cases}
        \1  & \mathit{ if }  |\O(p)|  > |\O(\neg p)|\\
        \0    & \mathit{ if }  |\O(\neg p)| > |\O(p)| \\
       \ast   & \mathit{ otherwise }   \\
\end{cases}
\end{align}
We will refer to this rule simply as `majority'. Majority can be thought of as a quota rule. Quota rules in binary aggregation with abstentions are of the following form: accept when the proportion of \emph{non-abstaining} individuals who accept is above the acceptance-quota; reject when the proportion of \emph{non-abstaining} individuals who reject is above the rejection-quota; and abstain otherwise:\footnote{There are several ways to think of quota rules with abstentions. Instead of a quota being a proportion of non-abstaining agents, one could for instance define rules with absolute quotas instead: accept when at least $n$ agents accept, independently of how many agents do not abstain. In practice, voting rules with abstention are often a combination of those two ideas: accept an issue if a big enough proportion of the population does not abstain, and if a big enough proportion of those accept it.}

\begin{definition}[Quota rules] \label{def:quota}
Let $\S$ be a BA structure.  
A {\em quota rule} (for $\S$) is defined as follows, for any issue $p\in\Atoms$, and any opinion profile $\O\in (\AB_{c})^\N$:
\begin{align}\label{eq:quotarules}
F(\O)(p)=  
\begin{cases}
        \1  & \mbox{ if }  |\O(p)|  \geq \left\lceil q_\1(p) \cdot |\O(\pm p)| \right\rceil\\
        \0  & \mbox{ if }  |\O(\neg p)| \geq \left\lceil q_\0(p) \cdot |\O(\pm p)| \right\rceil  \\
       \ast & \mbox{ otherwise } \\
\end{cases}
\end{align} \label{eq:quota}
where $\lceil \cdot \rceil$ is the cealing function. And, for $x \in \set{\0, \1}$, $q_x$ is a function $q_x: \Atoms \to (0,1] \cap \mathbb{Q}$ assigning a positive rational number smaller or equal to $1$ to each issue, and such that, for each $p \in \Atoms$:
\begin{align}
%0 < q_x\leq 1 \label{eq:nontrivial}, \\ 
q_x(p) > 1 - q_{(\1 - x)}(p). \label{eq:constraint}
\end{align}
A quota rule is called: {\em uniform} if, for all $p_i,p_j \in \Atoms$, $q_x(p_i) = q_x(p_j)$;
it is called {\em symmetric} if, for all $p \in \Atoms$, $q_\1(p) = q_\0(p)$.%%\footnote{The definition of symmetric rule could be relaxed somewhat to: for all $p\in\Atoms$, and all $\O\in\AB$, $\lceil q_\1(p) \cdot |\O(\pm p)| \rceil =  \lceil q_\0(p) \cdot |\O(\pm p)| \rceil$.}
%If $q_\1 = q_\0$ then the quota rule is called {\em symmetric}.
\end{definition}
Notice that the definition excludes trivial quota.\footnote{Those are quotas with value $0$ (always met) or $>1$ (never met). Restricting to non-trivial quota is not essential but simplifies our exposition.} 
It should also be clear that, by the constraint in \eqref{eq:constraint}, Definition \ref{def:quota} defines an aggregator of type $(\AB_{c})^\N \to \AB$ as desired.\footnote{What needs to be avoided here is that both the acceptance and rejection quota are set so low as to make the rule output both the acceptance and the rejection of the issue.} Notice finally that if the rule is symmetric, then \eqref{eq:constraint} forces $q_x(p) > \frac{1}{2}$, for any given $p$.

\begin{example} \label{example:maj}
The majority rule \eqref{eq:majast} is a uniform and symmetric quota rule where $q_\1$ and $q_\0$ are set to meet the equation $\lceil q_\1(p) \cdot |\O(\pm p)| \rceil = \lceil q_\0(p) \cdot |\O(\pm p)| \rceil = \left\lceil \frac{|\O(\pm p)| + 1}{2}\right\rceil$, for any issue $p$ and profile $\O$. This is achieved by setting the quota as $\frac{1}{2}<q_\1(p),q_\0(p) \leq \frac{1}{2}+\frac{1}{|N|} = \frac{|N|+1}{2|N|}$, for each issue $p$. More precisely one should therefore consider $\maj$ as a class of quota rules yielding the same collective opinions.
\end{example}

\begin{example}
The uniform and symmetric unanimity rule is defined by setting $q_\1 = q_\0 = 1$. A uniform but asymmetric variant of unanimity can be obtained by setting $q_\1 = 1$ and $q_\0 = \frac{1}{|\N|}$.
\end{example}

%Let us note an important difference between quota rules in binary aggregation with abstentions vs. without abstentions. In a framework without abstentions quota rules are normally %defined by a unique acceptance quota $q^\1$, the rejection quota being uniquely determined as $q^\0 = 1 - q^\1$. As a consequence, the majority rule, when $|N|$ is odd, is the only %unbiased quota rule in the standard framework. This is no longer the case when abstentions are considered. A novel characterization of the majority rule will be given in Section %%%%\ref{subsec:char.quota}.

%%%%%%%%%%%%%%%%%%%%%%%%%%%%%%%%

\subsection{Properties of Agendas and Aggregators} \label{subsec:axiomatic}

%Given an aggregation structure $\S$ with the set of %issues $\I$ let $\pm \I = \set{p \mid p \in \I} \cup %\set{\neg p \mid p \in \I}$ be called the {\em agenda} %of $\S$.

\begin{definition}[simple/evenly negatable agenda]
An agenda $\pm \I$ is said to be {\em simple} if there exists no set $X \subseteq \pm \Atoms$  such that: $|X|\geq 3$, and  $X$  is minimally $\gamma$-inconsistent, that is:
\begin{itemize}
\item $X$ is inconsistent with $\gamma$
\item For all  $Y\subset X$, $Y$ is consistent with $\gamma$ (or, $\gamma$-consistent).
\end{itemize}
An agenda is said to be {\em evenly negatable} if there exists a minimal $\gamma$-inconsistent set $X \subseteq \pm \Atoms$ such that for a set $Y \subseteq X$ of even size, $X\backslash Y \cup \set{\neg p \mid p \in Y}$ is $\gamma$-consistent. It is said to be {\em path-connected} if there exists $p_1, \ldots, p_n \in \pm \I$ such that $p_1 \models^c p_2, \ldots, p_{n-1} \models^x p_n$ where $p_i \models^c p_{i+1}$ (conditional entailment) denotes that there exists $X \subseteq \pm\I$, which is $\gamma$-consistent with both $p_i$ and $\neg p_{i+1}$, and such that $\set{p} \cup X \cup \set{\gamma}$ logically implies $p_{i +1}$. 
\end{definition}
We refer the reader to \cite[Ch. 2]{Grossi_2014} for a detailed exposition of the above rather technical conditions. We provide just a simple illustrative example here.

\begin{example}
Let $\I = \set{p, q, r}$ and let $\gamma = (p \wedge q ) \rightarrow r$. $\pm \I$ is not simple. The set $\set{p, q, \lnot r} \subseteq \pm\Atoms$ is inconsistent with $\gamma$, but none of its subsets is. Let now $\I = \set{p, q, r}$ and let $\gamma = (r \limp q) \land (q \limp p)$. In this case, where issues are ordered by logical entailment, each minimally $\gamma$-inconsistent set is of size $2$, and the agenda is therefore simple. The trivial example of simple agenda is where $\gamma = \top$, and the issues are therefore logically independent. 
\end{example}

We proceed by recalling some well-known properties of aggregators from the judgment and binary aggregation literatures, adapted to the setting of aggregation with abstention:\footnote{Such adaptation is, in many cases, non-trivial.}

\begin{definition} \label{def:properties}
Let $\S$ be an aggregation structure. An aggregator $F: (\AB_{c})^\N \to \AB$ is said to be:
\begin{description}
\item[unanimous] iff for all $p\in \Atoms$, for all profiles $\O$ and all $x\in\{0,1,\ast\}$: if for all $i\in \N, O_i(p) = x$, then $F(\O)(p)= x $. I.e., if everybody agrees on a value, that value is the collective value.
\item[anonymous] iff for any bijection $\mu: \N\rightarrow\N$, $F(\O)=F(\O^\mu)$, where $\O^\mu = \tuple{O_{\mu(1)}, \ldots, O_{\mu(n)}}$. I.e., permuting opinions among individuals does not affect the output of the aggregator.
%\item[$p$-dictatorial] iff there exists $i \in \N$ (the {\em $p$-dictator}) s.t. for any profile $\O$, and all $x \in \set{\0,\1}$, $O_i(p) = x$ iff $F(\O)(p) = x$.
%I.e., there exists an agent whose definite opinion determines the group's definite opinion on $p$. If $F$ is $p$-dictatorial, with the same dictator on all issues $p \in \Atoms$, then it is %called {\bf dictatorial}.
\item[$p$-oligarchic] iff there exists $C \subseteq \N$ (the {\em $p$-oligarchs}) s.t. $C\neq\emptyset$ and for any profile $\O$, and any value $x \in \set{\0,\1}$, $F(\O)(p) = x$ iff $O_i(p) = x$ for all $i\in C$.
I.e., there exists a group of agents whose definite opinions always determine the group's definite opinion on $p$. If $F$ is $p$-oligarchic, with the same oligarchs on all issues $p \in \Atoms$, then it is called {\bf oligarchic}.
\item[monotonic] iff, for all $p\in \Atoms$ and all $i\in\N$: 
for any profiles $\O, \O'$, if $\O =_{-i} \O'$: (i) if $O_i(p)\neq \1$ and $O'_i(p)\in\{\1,\ast\}$, then: if $F(\O)(p)=\1$, then $F(\O')(p)=\1$; and (ii) if $O_i(p)\neq \0$ and $O'_i(p)\in\{\0,\ast\}$, then: if $F(\O)(p)=\0$, then $F(\O')(p)=\0$. I.e., increasing support for a definite collective opinion does not change that collective opinion.
\item[independent] iff, for all $p\in \Atoms$, for any profiles $\O, \O'$: if for all $i\in \N, O_i(p) = O'_i(p)$, then $F(\O)(p)=F(\O')(p)$. I.e., the collective opinion on each issue is determined only by the individual opinions on that issue.
\item[neutral] iff, for all $p,q \in \Atoms$, for any profile $\O$: if for all $i\in\N$, $O_i(p)=O_i(q)$, then $F(\O)(p)=F(\O)(q)$. I.e., all issues are aggregated in the same manner.
%\item[systematic] iff it is neutral and independent. I.e., the collective opinion on issue $p$ depends only on the individual opinions on this issue.
\item[responsive] iff for all $p\in\Atoms$, there exist profiles $\O, \O'$  such that $F(\O)(p)=\1$ and $F(\O')(p)=\0$. I.e., the rule allows for an issue to be accepted for some profile, and rejected for some other. 
\item[unbiased] iff for all $p \in \Atoms$, for any profiles $\O, \O'$ : if for all $i\in\N$, $O_i(p)= \1$ iff $O'_i(p)=\0$ (we say that $\O'$ is the ``reversed'' profile of $\O$), then $F(\O)(p)=\1$ iff $F(\O')(p)=\0$. I.e., reversing all and only the individual opinions on $p$ (from acceptance to rejection and from rejection to acceptance) results in reversing the collective opinion on $p$. 
\item[rational] iff for any profile $\O$, $F(\O)$ is consistent and closed. I.e., the aggregator preserves the constraints on individual opinions.
\end{description}
\end{definition}
Majority is unanimous, anonymous, monotonic, independent, neutral, responsive and unbiased, but it is not rational in general, as witnessed by well-known judgment aggregation paradoxes (cf. \cite{Grossi_2014}).

Finally, let us also define the following property. The {\bf undecisiveness} of an aggregator $F$ on issue $p$ for a given aggregation structure is defined as the number of profiles which result in collective abstention on $p$, that is:
\begin{align}
u(F)(p) & = |\set{\O \in \AB_c \mid F(\O)(p) = \ast}|.
\end{align}

%%%%%%%%%%%%%%%%%%%%%%%%%%%%%%%%%%%%%%%%%

\subsection{Some Results}\label{subsec:char.quota}

Aggregation by majority is collectively rational under specific assumptions on the aggregation constraint:

\begin{proposition} \label{fact:rat}
Let $\S$ be a BA structure with a simple agenda. Then $\maj$ is rational.
\end{proposition}

May's theorem \cite{May_1952} famously shows that for preference aggregation, the majority rule is in fact the \emph{only} aggregator satisfying a specific set of desirable properties. A corresponding characterization of the majority rule is given in standard judgment aggregation (without abstentions): when the agenda is simple, the majority rule is the only aggregator which is rational, anonymous, monotonic and unbiased \cite[Th. 3.2]{Grossi_2014}. We give below a novel characterization theorem, which takes into account the possibility of abstentions both at the individual and at the collective level. To the best of our knowledge this is the first result of this kind in the literature on judgment and binary aggregation with abstention.  

\smallskip

We first prove the following lemma:
\begin{lemma} \label{lemma:min}
Let $F$ be a uniform and symmetric quota rule for a given $\S$. The following holds:
$\frac{1}{2}< q_\1 = q_0 \leq \frac{|N|+1}{2|N|}$ if and only if $F = \arg\min_G u(G)(p)$, for all $p \in \Atoms$. 
\end{lemma}
That is, the quota rule(s) corresponding to the majority rule (Example \ref{example:maj}) is precisely the rule that minimizes undecisiveness.

We can now state and prove the characterization result:
\begin{theorem}\label{thm:quotarules}
Let $F: (\AB_{c})^\N \to \AB$ be an aggregator for a given $\S$. The following holds: 
\begin{enumerate}
\item $F$ is a quota rule if and only if it is anonymous, independent, monotonic, and responsive; 
\item $F$ is a uniform quota rule if and only if it is a neutral quota rule;
\item $F$ is a symmetric quota rule if and only if it is an unbiased quota rule; 
\item $F$ is the majority rule $\maj$ if and only if it is a uniform symmetric quota rule which minimizes undecisiveness. 
\end{enumerate}
\end{theorem}
By the above theorem and Proposition~\ref{fact:rat}, it follows that, on simple agendas, majority is the only rational aggregator which is also responsive, anonymous, systematic and monotonic.

\smallskip 
%%%%%%%%%%%%%%%%%%%%%%%%%%%%%%%%%%%%%%%%%%%%%%%%%%%%%%

%\subsection{Impossibility in Binary Aggregation with Abstention}

We conclude by recollecting a well-known impossibility result concerning binary aggregation with abstentions:
\begin{theorem}[\cite{Dokow_2010,Dietrich_2007}]\label{thm:imp.agg.abst}
Let $\S$ be a BA structure whose agenda is path connected and evenly negatable. Then if an aggregator $F: (\AB_{c})^\N \to \AB$ is independent, unanimous and collectively rational, then it is oligarchic.
\end{theorem}
%We will use this result to illustrate how impossibility results from binary aggregation with abstentions apply to delegable proxy voting on binary issues.

%%%%%%%%%%%%%%%%%%%%%%%%%%

\section{Binary Liquid Democracy
% in Binary Aggregation with Abstentions
 } \label{sec:proxy}

In binary aggregation with delegable proxy, agents either express an acceptance/rejection opinion or \emph{delegate} the expression of such an opinion to another agent. The section models and studies this type of voting as a form of binary aggregation function. 
%The model allows us to frame two issues with the system.

\subsection{Proxy Opinions, Profiles and Delegation Graphs}

Let a BA structure $\S$ be given and assume for now that $\gamma = \top$, that is, all issues are logically independent. An opinion $O: \Atoms \to \set{\0,\1} \cup \N$ is an assignment of either a truth value or another agent to each issue in $\Atoms$, such that $O_i(p) \neq i$ (that is, self-delegation is not an expressible opinion).
We will later also require proxy opinions to be individually rational, in a precise sense (Section \ref{sec:indirat}). 
For simplicity we are assuming that abstention is not a feasible opinion in proxy voting, but such assumption can be easily lifted in what follows.

We call functions of the above kind {\em proxy opinions} to distinguish them from standard (binary) opinions, and we denote by $\mathcal{P}$ the set of all proxy opinions, $\mathcal{P}_c$ the set of all individually rational proxy opinions (as defined later in Section \ref{sec:indirat}). Finally, $\mathcal{P}^\N$ denotes the set of all profiles of proxy opinions, which we call, {\em proxy profiles}.  

%\subsection{Aggregation of Proxy Opinions}

%\subsubsection{Delegation Graphs}

\medskip

Each proxy profile $\O$ induces a {\em delegation graph} $G^\O = \langle \N, \set{R_p}_{p \in \I}\rangle$ where for $i, j \in \N$:
\begin{align}
iR_p j & \IFF
\left\{
\begin{array}{ll}
O_i(p) = j & \mbox{if $i \neq j \in \N$} \\
O_i(p) \in \set{\0, \1} & \mbox{otherwise}
\end{array}
\right.
\end{align}
The expression $iR_pj$ stands for ``$i$ delegates her vote to $j$ on issue $p$''. 
Each $R_p$ is a so-called functional relation. It corresponds to the graph of an endomap on $\N$. So we will sometimes refer to the endomap $r_p: \N \to \N$ of which $R_p$ is the graph. Relations $R_p$ have a very specific structure and can be thought of as a set of (converging) trees whose roots all belong to cycles (possibly loops).

The weight of an agent $i$ w.r.t. $p$ in a delegation graph $G^\O$ is given by its indegree with respect to $R^*_p$ (i.e., the reflexive and transitive closure of $R_p$):\footnote{
We recall that the reflexive transitive closure $R^*$ of a binary relation $R \subseteq \N^2$ is the smallest reflexive and transitive relation that contains $R$.
} 
 $w^\O_i(p) = \left|\set{j \in \N \mid j R^*_p i} \right|$. The weight of a coalition $C \subseteq \N$ is defined naturally as $w^\O_C(p) = \sum_{i \in C} w^\O_i(p)$.
 This definition of weight makes sure that each individual carries the same weight, independently of the structure of the delegation graph. Alternative definitions of weight are of course possible.
 % and we will come back to this issue later.
 %\footnote{See also footnote \ref{footnote:weight} below.} 
 
For all $p \in \I$, we consider the function $g_p: \N \rightarrow \wp(\N)$ defined as
$
g_p(i) = \set{j \in \N \mid i R^*_p j \AND jR_p j}.
$
The function associates to each agent $i$ (for a given issue $p$), the (singleton consisting of the) last agent reachable from $i$ via a path of delegation on issue $p$, when it exists (and $\emptyset$ otherwise). Slightly abusing notation we will use $g_p(i)$ to denote an agent, that is, the {\em guru} of $i$ over $p$ when $g_p(i) \neq \emptyset$. If $g_p(i) = \set{i}$ we call $i$ a {\em guru} for $p$. Notice that $g_p(i) = \set{i}$ iff $r_p(i) = i$, i.e., $i$ is a guru of $p$ iff it is a fixpoint of the endomap $r_p$.
 
If the delegation graph $G^\O$ of a proxy profile $\O$ is such that, for some $R_p$, there exists no $i \in N$ such that $i$ is a guru of $p$, we say that graph $G^\O$ (and profile $\O$) is {\em void} on $p$. Intuitively, a void profile is a profile where no voter expresses an opinion, because every voter delegates her vote to somebody else.

Given a BA structure $\S$, a proxy aggregation rule (or proxy aggregator) for $\S$ is a function $\pv:\mathcal{P}^\N\to\AB$ that maps every proxy profile to one collective incomplete opinion. As above, $\pv(\O)(p)$ denotes the outcome of the aggregation on issue $p$.

\subsection{Proxy Aggregators}

The most natural form of voting via delegable proxy is a proxy version of the majority rule we discussed in Section \ref{sec:preliminaries}:\footnote{On the importance of majority decisions in the current implementation of liquid democracy by Liquid Feedback cf. \cite[p.106]{liquid_feedback}.}
\begin{align}\label{eq:proxymajast}
\pv_{\maj}(\O)(p) =
\begin{cases}
 \1 &  \mbox{ if }  \sum_{i \in \O(p)} w^\O_i(p) > \sum_{i \in \O(\neg p)}w^\O_i(p) \\
 \0 &  \mbox{ if }  \sum_{i \in \O(\neg p)} w^\O_i(p) > \sum_{i \in \O(p)}w^\O_i(p) \\
 \ast & \mbox{ otherwise }
\end{cases}
\end{align}
Again, the notation $\O(p)$ (resp., $\O(\neg p)$) denotes the set of voters accepting (resp., rejecting) $p$ in proxy profile $\O$.
Intuitively, an issue is accepted by proxy majority in profile $\O$ if the sum of the weights of the agents who accept $p$ in $\O$ exceeds the majority quota, it is rejected if the sum of the weights of the agents who reject $p$ in $\O$ exceeds the majority quota, and it is undecided otherwise. Note that $\sum_{i \in \O(p)} w^\O_i(p) = |\{i \in \N | O_{g_i}(p)= \1 \}|$ (and similarly for $\neg p$), that is, the sum of the weights of the gurus accepting (rejecting) $p$ is precisely the cardinality of the set of agents whose gurus accept (reject) $p$. 

It should be clear that for any quota rule $F: \AB_c \to \AB$ a proxy variant $\pv_F$ of $F$ can be defined via an obvious adaptation of \eqref{eq:proxymajast}.

%%%%%%%%%%%%%%%%%%%%%%%%%

\section{Analysis and Extensions} \label{sec:extensions}

In this section we provide an analysis of liquid democracy by highlighting two issues---the failure of rationality in ballots under delegable proxy voting, and the occurrence of delegation cycles---and by embedding it in the theory of binary aggregation with abstentions presented in Section \ref{sec:preliminaries}. We also advance proposals for simple modifications of the delegable proxy voting method in order to address the issues we identify.

\subsection{Individual and Collective Rationality} \label{sec:indirat}

In our discussion so far we have glossed over the issue of logically interdependent issues and collective rationality. The reason is that under the delegative interpretation of liquid democracy developed in the previous sections individual rationality itself appears to be a more debatable requirement than it normally is in classical aggregation. 

A proxy opinion $O_i$ is {\em individually rational} if the set of formulas
\begin{align}
\set{\gamma} \cup \set{p \in \I \mid O_{g_p(i)}(p) = \1} \cup \set{\neg p \in \I \mid O_{g_p(i)}(p) = \0} \label{eq:ir}
%\set{\gamma} \cup \set{p \in \I \mid O_i(p) = \1 \mbox{   or   } O_{g_p(i)}(p) = \1} \cup %\set{\neg p \in \I \mid O_i(p) = \0 \mbox{   or   } O_{g_p(i)}(p) = \0} \label{eq:ir}
\end{align}
is satisfiable (consistency), and if whenever \eqref{eq:ir} entails $\pm p$, then $\pm p$ belongs to it (closure).
%\footnote{Cf. the definition of individual opinions in Section \ref{sec:preliminaries}.} 
That is, the integrity constraint $\gamma$ is consistent with $i$'s opinion on the issues she does not delegate on, and the opinions of her gurus (if they exist), and those opinions, taken together, are closed under logical consequence. 

The consistency and closure of \eqref{eq:ir} capture a highly idealized way of how delegation works: voters are assumed to be able to check or monitor how their gurus are going to vote, and always modify their delegations if an inconsistency arises. So the constraint appears highly unrealistic under a delegative interpretation of liquid democracy. Aggregation via delegable proxy has at least the potential to represent individual opinions as irrational (inconsistent and/or not logically closed). 
%---that is why, intuitively, she delegates her vote in the first place.

The assumption of individual rationality for proxy opinions, however, is needed in order to establish variants of known binary aggregation results for the case of liquid democracy, to which we turn now.

\subsection{Embedding}

Having defined individual rationality in the previous section, it is possible now to study embeddings from proxy voting to standard aggregation, and vice versa. 

Aggregation in liquid democracy---as conceived in \cite{liquid_feedback}---should satisfy the principle that the opinion of every voter, whether expressed directly or through proxy, should be given the same weight.\footnote{
``[\ldots] in fact every eligible voter has still exactly one vote [\ldots] unrestricted transitive delegations are an integral part of Liquid Democracy. [\ldots] Unrestricted transitive delegations are treating delegating voters and direct voters equally, which is most democratic and empowers those who could not organize themselves otherwise'' \cite[p.34-36]{liquid_feedback}
}
In other words, this principle suggests that aggregation via delegable proxy should actually be `blind' for the specific type of delegation graph arising. Making this more formal, we can think of the above principle as suggesting that the only relevant content of a proxy profile is its translation into a standard opinion profile (with abstentions) via a function $t: \mathcal{P} \to \AB$ defined as follows: for any $i \in \N$ and $p \in \I$, $t(O_i(p)) = O_{g_p(i)}$ if $ g_p(i) \neq \emptyset$ (i.e., if $i$ has a guru for $p$), and $t(O_i(p)) = \ast$ otherwise. Clearly, if we assume proxy profiles to be individually rational, the translation will map proxy opinions into individually rational (consistent and closed) incomplete opinions. By extension, we will denote by $t(\O)$ the incomplete opinion profile resulting from translating the individual opinions of a proxy profile $\O$.

\medskip

The above discussion suggests the definition of the following property of proxy aggregators: a proxy aggregator $\pv$ has the {\bf one man--one vote property} (or is a one man--one vote aggregator) if and only if $pv = t \circ F$ for some aggregator $F: \AB_c \to \AB$ (assuming the individual rationality of proxy profiles).\footnote{Not every proxy aggregator satisfies the one man--one vote property. By means of example, consider an aggregator that uses the following notion of weight accrued by gurus in a delegation graph. The weight $w(i)$ of $i$ is $\sum_{j \in R^*(i)} \frac{1}{\ell(i,j)}$ where $\ell(i,j)$ denotes the length of the delegation path linking $j$ to $i$. This definition of weight is such that the contribution of voters decreases as their distance from the guru increases. Aggregators of this type are studied in \cite{Boldi_2011}. \label{footnote:weight}} The class of one man--one vote aggregators can therefore be studied simply as the concatenation $t \circ F$ where $F$ is an aggregator for binary voting with abstentions,
as depicted in Figure \ref{figure:manvote} (left).

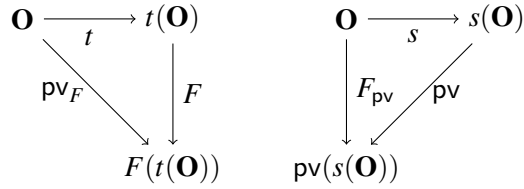
\begin{figure}
\begin{center}
\begin{tikzpicture}
\node(P) at (0,0) {$\O$};
\node(O) at (2,0) {$t(\O)$};
%\node(Q) at (4,0) {$s(t(\O))$};
\node(F) at (2, -2) {$F(t(\O))$};
\draw[->] (P) -- node[below]{$t$}     (O); 
\draw[->] (P) -- node[left]{$\pv_F$}  (F);
%\draw[->] (Q) -- node[right]{$\pv_F$}  (F);
\draw[->] (O) -- node[right]{$F$}  (F);
%\draw[->] (O) -- node[below]{$s$}  (Q);
%\draw[->] (Q) -- node[above]{$t$}  (O);
\end{tikzpicture}
\hspace{0.5cm}
\begin{tikzpicture}
%\node(P) at (0,0) {$\O$};
\node(O) at (2,0) {$\O$};
\node(Q) at (4,0) {$s(\O)$};
\node(F) at (2, -2) {$\pv(s(\O))$};
%\draw[->] (P) -- node[below]{$t$}     (O); 
%\draw[->] (P) -- node[left]{$\pv_F$}  (F);
\draw[->] (Q) -- node[right]{$\pv$}  (F);
\draw[->] (O) -- node[right]{$F_\pv$}  (F);
\draw[->] (O) -- node[below]{$s$}  (Q);
%\draw[->] (Q) -- node[above]{$t$}  (O);
\end{tikzpicture}
\end{center}
\caption{Embeddings to and from binary aggregation.}
\label{figure:manvote}
\end{figure}

\begin{example}
Proxy majority $\pv_\maj$ \eqref{eq:proxymajast} is a one man--one vote rule aggregator. It is easy to check that, for any proxy profile $\O$: $\pv_{\maj}(\O) = \maj(t(\O))$. 
%The same holds for proxy dictatorship \eqref{eq:proxyd}. It is easy to see that proxy dictatorship $\pv_d$ is such that for any proxy profile $\O$: $\pv_{d}(\O) = d(t(\O))$, where $d$ is %the standard dictatorship (of $d \in \N$).
\end{example}

%\subsubsection{Characterization and Impossibility Results}

It follows that for every proxy aggregator $\pv_F = t \circ F$ the axiomatic machinery developed for standard aggregators can be directly tapped into. 
Characterization results then extend effortlessly. In particular, Theorem~\ref{thm:quotarules} implies the following: 
%the following characterization of proxy majority:
\begin{fact}[Characterization of proxy majority]\label{thm:proxy.quotarules}
A one man--one vote proxy aggregator $\pv = t \circ F$ for a given $\S$ is proxy majority $\pv_\maj$ iff $F$ is anonymous, independent, monotonic, responsive, neutral and minimizes undecisiveness.
\end{fact}
%A consequence of Fact~\ref{thm:proxy.quotarules},
%By Theorem \ref{thm:quotarules},    
%is that proxy majority is the only rational aggregator which is anonymous, independent, monotonic, %responsive, neutral and minimizes undecisiveness. 
The fact may well be considered as a theoretical argument in favor of the use of proxy majority in aggregation with delegable proxy as currently done, for instance, in the Liquid Feedback platform.

\smallskip

Similarly, we can study an embedding of standard aggregation into voting with delegable proxy. For example, we can define a function $s: \mathcal{O}^\ast_c \to \mathcal{P}_c$ from opinion profiles to individually rational proxy profiles as follows. For a given opinion profile $\O$, and issue $p$ consider the set $\set{i\in \N \mid O_i(p) = \ast}$ of individuals that abstain in $\O$ and take an enumeration $\sigma: \set{i\in \N \mid O_i(p) = \ast} \to \set{1, \ldots, m}$ of its elements, with $m = |\set{i\in \N \mid O_i(p) = \ast}|$. The function is defined as follows: for any $i \in \N$ and $p \in \I$, $s(O_i(p)) = O_i(p)$ if $O_i(p) \in \set{\0,\1}$, $s(O_i(p)) = (\sigma(i)+1) \bmod m$, otherwise.\footnote{Notice that since self-delegation (that is, $O_i(p) = i$) is not feasible in proxy opinions, this definition of $s$ works for profiles where, on each issue, either nobody abstains or at least two individuals abstain. Clearly, a dummy abstaining voter can then be added in profiles where only one individual abstains.} 
A translation of this type allows to think of standard aggregators $F: \AB_c \to \AB$ as the concatenation $s \circ \pv$, for some proxy aggregator $\pv$, as in Figure \ref{figure:manvote} (right). The following impossibility result for aggregation with delegable proxy voting can then be obtained as a direct consequence of Theorem \ref{thm:imp.agg.abst}:

\begin{fact} \label{thm:imp.proxy}
Let $\S$ be such that its agenda is path connected and evenly negatable. For any proxy aggregator $\pv$, if $s \circ \pv$ is independent, unanimous and collectively rational, then it is oligarchic.
\end{fact}

%%%%%%%%%%%%%%%%%%%%%%%%%%%%%%%%%%

\subsubsection{Cycles and Abstentions} \label{sec:proxyabs}

Proxy aggregators rely on the existence of gurus in the underlying delegation graphs. If the delegation graph $R_p$ on issue $p$ contains no guru, then the aggregator has access to no information in terms of who accepts and who rejects issue $p$. To avoid bias in favor of acceptance or rejection, such situations should therefore result in an undecided collective opinion. That is for instance the case of $\pv_\maj$. However, such situations may well be considered problematic, and the natural question arises therefore of how likely they are, at least in principle.
\begin{proposition} \label{fact:cycles}
Let $\S$ be a BA structure where $\gamma = \top$ (i.e., issues are independent) and fix an issue $p$.
%, and consider a unanimous, one man--one vote proxy aggregator.
If each proxy profile is equally probable (impartial culture assumption), then the probability that a given proxy profile $\O$ is such that $t(\O)$ is a profile in which every voter abstains tends to $\frac{1}{e^2}$ as $n$ tends to infinity.
%results in collective abstention tends to $\frac{1}{e^2}$ as $n$ tends to infinity.
%the delegation graph $R_p$ has no gurus tends to $\frac{1}{e^2}$ as $n$ tends to infinity. 
%If each delegation graph $R_p$ on $p$ is %equally probable, then the probability %that $R_p$ has no gurus tends to %$\frac{1}{e}$ as $n$ tends to infinity.
\end{proposition}
It follows that for unanimous and one man--one vote proxy aggregators, asymptotically, there is a considerable chance that a profile results in collective abstention. Now contrast this with the probability that all agents abstain on an issue when each voter either expresses a $\1$ or $\0$ opinion or abstains (that is, the binary aggregation with abstentions setting studied earlier). In that case the probability that everybody abstains, and therefore the profile is void, clearly tends to $0$ as $n$ tends to infinity.

Proposition~\ref{fact:cycles} should obviously not be taken as a realistic estimate of the effect of cycles on collective abstention, moreover concrete implementations of delegable proxy voting may be designed to detect and resolve cycles (cf. \cite{yamakawa07toward,kling15voting}). Ultimately, theoretical (e.g., game theoretic) models of delegation behavior in voters or, ideally, 
election data should be used to assess whether delegation cycles ever lead large parts of the electorate to effectively lose representation in the aggregation mechanism. Still, the link we highlight between delegable proxy and collective abstention is, to the best of our knowledge, novel and has escaped so far recognition within the liquid democracy literature.\footnote{Delegation cycles are normally criticized for the wrong reason, that is, the fact that hey may be interpreted as to lead to an infinite accrual of voting power: ``The by far most discussed issue is the so-called circular delegation
problem. What happens if the transitive delegations lead to
a cycle, e.g. Alice delegates to Bob, Bob delegates to Chris, and
Chris delegates to Alice? Would this lead to an infinite voting
weight? Do we need to take special measures to prohibit such a
situation? In fact, this is a nonexistent problem: A cycle only exists as long as there is no activity in the cycle in which case the cycle has no effect. As already explained [\ldots], as soon as somebody casts a vote, their (outgoing) delegation will be suspended. Therefore, the cycle naturally disappears before it is used. In our example: If Alice and Chris decide to vote, then Alice will no longer delegate to Bob, and Chris will no longer delegate to Alice [\ldots]. If only Alice decides to vote, then only Alice's delegation to Bob is suspended and Alice would use a voting weight of 3. In either case the cycle is automatically resolved and the total voting weight used is 3.'' \cite[Section 2.4.1]{liquid_feedback} Cf. \cite{Behrens15}. We agree that the alleged accrual of infinite voting power is immaterial. However, the fact that the occurrence of a cycle leads to the loss of representation of the voters in the cycle---and of those delegating to them---does not seem to have yet been acknowledged.
}

%%%%%%%%%%%%%%%%%%%%%%%  proxy

\subsection{Delegable Proxy with Default Values} \label{sec:default}

Motivated by the above analysis, we outline a simple modification of voting via delegable proxy, which requires agents to always submit a substantive opinion on the issues, and at the same time indicate a trustee. In this view, an opinion (called {\em proxy opinion with default}) is therefore a function $O_i: \Atoms \to (\set{\0,\1} \times \N)$ assigning to every issue an acceptance or rejection value and, at the same time, an individual, which is to be considered the individual the vote is delegated to. Intuitively, each voter expresses an opinion but accepts that opinion to be overruled by the opinion of the individual she entrusts. Note that such individual may well be the voter herself (e.g., $O_i(p) = (\1, i)$). We refer to profiles of such opinions as {\em proxy profiles with default}.  

Let $\C_\O(p) = \set{C \subseteq \N \mid C ~\mbox{is a $R_p$-cycle}~\AND |\O_C(p)| > |\O_C(\neg p)|}$ denote the set of cycles of the delegation graph $R_p$ such that among the agents in the cycle there exists a majority accepting $p$. The set $\C_\O(\neg p)$ is defined in the symmetric way. Now define proxy majority as an aggregator for profiles of proxy opinions with default values:
\begin{align}\label{eq:default}
\hspace{-0.2cm} \pv'_{\maj}(\O)(p) =
\begin{cases}
 \1 &  \mbox{if }  \sum_{C \in \C_\O(p)} w^\O_C(p) > \sum_{C \in \C_\O(\neg p)}w^\O_C(p) \\
 \0 &  \mbox{if }  \sum_{C \in \C_\O(\neg p)} w^\O_C(p) > \sum_{C \in \C_\O(p)}w^\O_C(p) \\
 \ast & \mbox{otherwise }
\end{cases}
\end{align}
where, recall, $w^\O_C(p)$ is the cumulative weight (w.r.t. $R_p$) of the agents in $C$.
% $\O_C(p)$ (resp., $\O_C(\neg p)$) is the set of agents in $C$ that accept (resp., reject) $p$ in $\O$.
The intuition behind \eqref{eq:default} is to use each cycle, and not only loops (i.e., gurus), as sources of information for the proxy aggregator, by attributing to the individuals in a cycle the majority default opinion present in that cycle. 

As one might intuitively expect, this is enough to break the link between delegation cycles and group abstention we identified with Proposition \ref{fact:cycles}. To state the following result we need to adapt the translation function $t$ for proxy profiles, to a translation function $t'$ translating proxy profiles with default to opinion profiles with abstentions: for any $i \in \N$ and $p \in \I$, $t'(O_i(p)) = \maj(\O_C)(p)$ where $C$ is the cycle reachable from $i$ via $R_p$.

\begin{proposition} \label{fact:cycles_default}
Let $\S$ be a BA structure where $\gamma = \top$ (i.e., issues are independent) and fix an issue $p$.
%, and consider a unanimous, one man--one vote proxy aggregator.
If each proxy profile with default is equally probable (impartial culture assumption), then the probability that a given proxy profile with default $\O$ is such that $t'(\O)$ is a profile in which every voter abstains tends to $0$ as $n$ tends to infinity.
%\footnote{Here $t'$ is the natural adaptation of the translation function $t$ for proxy profiles with default to %standard binary profiles with abstentions.}
\end{proposition}

%%%%%%%%%%%%%%%%%%%%%%%%%%%%%%%%%%%%%%%%

\subsection{Individually Rational Delegable Proxy} \label{sec:diffusion}

Delegable proxy voting can also be studied from a different perspective. 
Imagine a group where, for each issue $p$, each agent copies the binary---$\0$ or $\1$---opinion of a unique trustee.\footnote{For simplicity, in this section we assume agents are therefore not allowed to abstain, although this is not a crucial assumption for the development of our analysis.} Imagine that this group does so repeatedly until all agents (possibly) reach a stable opinion. These new stable opinions can then be aggregated as the `true' opinions of the individuals in the group, for instance, via majority. The collective opinion of a group of agents, who either express a binary opinion or delegate it to another agent, is (for one man--one vote proxy aggregators) the same as the output obtained from a vote where each individual has to express a binary opinion but gets there by {\em copying} the opinion of her trustee (possibly the agent itself). 
In this perspective, aggregation via delegable proxy can be assimilated to a (stabilizing) process of opinion formation on delegation graphs. 

The above interpretation of liquid democracy is explicitly put forth in \cite{liquid_feedback}.\footnote{
``While one way to describe delegations is the transfer of voting weight to another person, you can alternatively think of delegations as automated copying of the ballot of a trustee.
While at assemblies with voting by a show of hands it is naturally possible to copy the vote of other people, in Liquid Democracy this becomes an intended principle'' \cite[p. 22]{liquid_feedback}.
} Under this `vote-copying' interpretation, the constraint on individual rationality---consistency and closure of \eqref{eq:ir}---is, arguably, more easily defendable: each agent will copy opinions coming from her trustees only if consistency and closure are preserved. 

\subsubsection{Boolean DeGroot Processes}

We briefly develop the above intuition, outlining an opinion diffusion model of delegable proxy which preserves individual rationality in a natural way.\footnote{As we will consider just binary opinions (without abstentions), the concept of individual rationality can be slightly simplified: requiring an opinion to be $\gamma$-consistent suffices as in the case of binary opinions without abstentions, consistency implies closedness.}
\begin{definition}
Fix a BA structure $\S = \tuple{\N,\Atoms,\gamma}$, a profile $\O \in \left(\set{\0,\1}^\I \right)^\N$ of $\gamma$-consistent binary opinions, and a delegation graph $G = \langle \N, \set{R_p}_{p \in \I}\rangle$. Consider the stream $\O^0, \O^1, \ldots, \O^n, \ldots$ of opinion profiles recursively defined as follows:
\begin{itemize}%[noitemsep]
\item Base: $\O_0 := \O$
\item Step: for all $i \in \N$, $p\in \Atoms$, 
\begin{align*}
\hspace{-0.2cm} O_i^{n+1}(p) := 
\left\{
\begin{array}{ll}
O^{n}_{R_p(i)}(p) & \mbox{if    }  \set{\gamma} \cup \set{p \in \I \mid O_{R_p(i)}(p) = \1} \cup \set{\neg p \in \I \mid O_{R_p(i)}(p) = \0} \\
& \mbox{is consistent} \\
%\bigwedge_{p \in \Atoms} O^{n}_{R_p(i)}(p) \wedge \gamma  \mbox{   consistent \& closed} \\
O_i^{n}(p) & \mbox{otherwise}
\end{array}
\right.
\end{align*}
\end{itemize}
where $G_p = \tuple{\N, R_p}$.
\end{definition}
When $\gamma$ is set to $\top$, the above defines $|\I|$ independent binary processes, one for each issue $p$. Each of such processes is a Boolean extremal case of a DeGroot stochastic process \cite{Degroot_1974} where opinions are binary, and each agent can trust one and at most one other agent.
When the constraint $\gamma$ is not a tautology, the definition guarantees that at each step individual opinions remain consistent with $\gamma$. We call processes defined by the above dynamics \emph{individually rational Boolean DeGroot processes} (in short, BDPs).\footnote{Other types of dynamics are of course possible. A recent systematic investigation of opinion diffusion on logically interdependent issues is \cite{Botan16}. For a broader study of Boolean DeGroot processes in the context of models of binary opinion diffusion on networks we refer the reader to \cite{christoff17stability}.}

\subsubsection{Stabilization}
We say that the stream of opinion profiles $\O^0, \O^1, \ldots, \O^n, \ldots$ {\em stabilizes} if 
there exists $n \in \mathbb{N}$ such that for all $m\in \mathbb{N}$, if $m\geq n$, then $\O^m = \O^n$. We call such profile the {\em limit profile}. A BDP that stabilizes can be thought of as an opinion transformation function \cite{List_2010} $f_G: \O \to \O$ turning an initial binary profile $\O$ into a new binary profile $f(\O)$ equal to the limit profile. In this view, individually rational proxy aggregation consists first in an opinion transformation, implemented through a BDP, and then the application of an aggregator (e.g., $\maj$) on the profile of transformed opinions $f(\O)$. A BDP that does not converge, can similarly be thought of as mapping the initial profile to a profile involving some level of abstention, where agents connected to some delegation cycle may not end up stabilizing and are therefore considered to abstain. We conclude by establishing conditions for individually rational Boolean DeGroot processes to stabilize. 

\begin{theorem}
\label{theorem:resistantsufficient}
%Let $G$ be a delegation graph, $\O$ be an opinion profile, and $\gamma$ a constraint. 
Fix a BA structure $\S = \tuple{\N,\Atoms,\gamma}$, a profile $\O$ of consistent (w.r.t. $\gamma$) binary opinions, and a delegation graph $G$.
Then the following holds: {\em if} for all $p\in \Atoms$, for all $C\subseteq\N$ such that $C$ is a cycle in $G_{p}$, and all $i,j\in C$: $O_i(p)=O_j(p)$, {\em then} the individually rational BDP (for $\O$, $G$ and $\gamma$) stabilizes in at most $k$ steps, where $k= \max\set{diam(G_p)|p\in \I}$.
\end{theorem}

\smallskip

When $\gamma = \top$, the opposite direction also holds, and one can obtain a characterization of the notion of stabilization for BDPs based on properties of the initial opinion profile and of the delegation graph.
\begin{theorem}
\label{theorem:opinion}
Fix a BA structure $\S = \tuple{\N,\Atoms,\gamma}$, a profile $\O$ of consistent (w.r.t. $\gamma$) binary opinions, and a delegation graph $G$, and let $\gamma = \top$. Then the following statements are equivalent:
\begin{enumerate}
\item The BDP (for $\O$ and G) stabilizes.
\item For all $p\in \Atoms$, there is no set of agents $S\subseteq\N$ such that: $S$ is a cycle in $G_{p}$ and there are two agents $i,j\in S$ such that $O_i(p)\neq O_j(p)$.
\end{enumerate}
\end{theorem}
A special case of Theorem \ref{theorem:opinion} is the case in which 
%for all $p\in\P$, 
$G_{p}$ contains no cycle of length $\geq 2$. In such case, a direct consequence of the theorem is that the process stabilizes from any profile. This is also a corollary of a known stabilization result for DeGroot processes (cf. \cite[p.233]{jackson08social}).

\section{Conclusions} \label{sec:conclusions}

The paper has shown how delegable proxy voting (liquid democracy) can be understood as an aggregator within the theory of binary aggregation with abstentions, for which we provided a novel characterization theorem of issue-wise majority (Theorem \ref{thm:quotarules}). This has allowed us to clarify the impact of cyclical delegations on individual and collective abstentions (Proposition \ref{fact:cycles}) and to suggest alternative aggregators requiring individuals to reveal a default opinion, which can be shown to better behave in the presence of delegation cycles (Proposition \ref{fact:cycles_default}). Finally we showed how delegable proxy interferes with individual rationality, a standard tenet of social choice theory. Also in this case we showed how liquid democracy could be adjusted---in the form of a stabilizing diffusion process---in order to preserve individual rationality (Theorem \ref{theorem:resistantsufficient}).

%%%%%%%%%%%%%%%%%%%%%%%%%%%%%%%%%%%%

\appendix

\section*{Proofs}

\begin{proof}[Proof of Proposition \ref{fact:rat}]
If the agenda $\pm\I$ is simple, then all minimally inconsistent sets have cardinality $2$, that is, are of the form $\set{\phi,\neg \psi}$ such that $\phi \models \neg \psi$ for $\phi,\psi \in \I$. W.l.o.g. assume $\phi = p_i$ and $\psi = p_j$. Suppose towards a contradiction that there exists a profile $\O$ such that $\maj(\O)$ is inconsistent, that is, $\maj(\O)(p_i) = \maj(\O)(p_j) = 1$, and $\phi \models \neg \psi$. By the definition of $\maj$ \eqref{eq:majast} it follows that $|\O(p_i)| > |\O(\neg p_i)|$ and $|\O(p_j)| > |\O(\neg p_j)|$. Since $p_i \models \neg p_j$ by assumption, and since individual opinions are consistent and closed, $|\O(\neg p_j)| \geq |\O(p_i)|$ and $|\O(\neg p_i)| \geq |\O(p_j)|$. From the fact that $|\O(p_i)| > |\O(\neg p_i)|$ we can thus conclude that $|\O(\neg p_j)|> |\O(p_j)|$. Contradiction.
\end{proof}

\begin{proof}[Proof of Lemma \ref{lemma:min}]
We establish the claim through a series of equivalences. 
Observe first of all that a uniform and symmetric quota rule $F$ is such that (a) $F = \arg\min_G u(G)(p)$, for all $p \in \Atoms$ if and only if, (b) for any $\O \in \AB_c$ and $p \in \Atoms$, $u(\O)(p) = \ast$ if and only if $\O(p) = \O(\neg p)$, that is, an even number of voters vote and the group is split in half. Now, (b) is the case if and only if, (c) the quota of $F$ are set in such a way that $\lceil q_\1(p)|\O(\pm p)| \rceil = \lceil q_\0(p) |\O(\pm p)| \rceil = \left\lceil \frac{|\O(\pm p)| + 1}{2}\right\rceil$ for any profile $\O$ and issue $p$. In turn (c) is the case if and only if, (d) the quota of $F$ are set as $\frac{1}{2}< q_\1(p) = q_0(p) \leq \frac{|N|+1}{2|N|}$, which are the quota defining $\maj$ (Example \ref{example:maj}).
%The claim is proven by the following series of equivalent statements. 
%(a) A uniform and symmetric quota rule $F$ has quota such that $\frac{1}{2}< q_\1 = q_0 \leq \frac{|N|+1}{2|N|}$. 
%(b) A uniform and symmetric quota rule $F$ has quota such that $\lceil q_\1(p)|\O(\pm p)| \rceil = \lceil q_\0(p) |\O(\pm p)| \rceil = \left\lceil \frac{|\O(\pm p)| + 1}{2}\right\rceil$ for any %profile $\O$ and issue $p$. 
%(c) For any $\O \in \AB_c$ and $p \in \Atoms$, $u(\O)(p) = \ast$ if and only if $\O(p) = \O(\neg p)$, that is, an even number of voters vote and the group is split in half. 
%By Definition of \ref{def:quota}, no symmetric quota can do better, so the latter is %equivalent to the claim that 
%(d) $F = \arg\min_G u(G)(p)$, for all $p \in \Atoms$.
\end{proof}

\begin{proof}[Proof of Theorem \ref{thm:quotarules}]
\smallskip
\fbox{Claim 1}
Left-to-right: Easily checked.
Right-to-left: Let $F$ be an anonymous, independent, monotonic, and responsive aggregator. By \textit{anonymity} and \textit{independence}, for any $p\in\Atoms$, and any $\O\in\AB_c$, the only information determining the value of $F(O)(p)$ are the integers $|\O(p)|$ and $|\O(\neg p)|$. 
%\item  by non-uselessness it is not the case that for all profiles $\O\in\AB$, $F(\O)(p)=*\}$. We proceed by case distinction:
By \textit{responsiveness}, there exists a non-empty set of profiles $S^\1=\{\O\in\AB|F(\O)(p)=\1\}$. Pick $\O$ to be any profile in $S^\1$ with a minimal value of $\frac{|\O(p)|}{|\O(\pm p)|}$ and call this value $q_\1$. Now let $\O'$ be any profile such that $\O'=_{-i}\O$ and $\frac{|\O'(p)|}{|\O'(\pm p)|}>q_\1$. This implies that $O_i(p)=\0$ and $O'_i(p)=\1$. By \textit{monotonicity}, it follows that $F(\O')(p)= \1$. 
By iterating this argument a finite number of times we conclude that whenever $\frac{|\O(p)|}{|\O(\pm p)|} \geq q_\1$, we have that $F(\O)(p)=\1$. 
Given that $q_\1$ was defined as a minimal value, we conclude also that if $F(\O)(p)=1$, then $\frac{\O(p)}{\O(p^\pm)}\geq q_\1$. The argument for $q_\0$ is identical.
%THAT WAS THE PART TO FIND THE ABSOLUTE QUOTA, ASSUMING WE DON'T START WITH A GIVEN NUMBER OF AGENTS: Fix $q_1$ to be the value such that for any aggregation structure $\S$ and any profile $\O$ in $\AB$: $\frac{|\O(p)|-1}{|\O(p^\pm)|}< q_\1 < q \in $: $F(\O)(p)=\1$ iff $\frac{|\O'(p)|}{|\O'(p^\pm|)}> q_\1$.  
%By \textit{responsiveness} again, there also exists a non-empty set of profiles $S^\0$ such %that $S^\0=\{\O\in\AB|F(\O)(p)=\0\}$. Take $\O$ to be one of the profiles in $S^\0$ with %minimal value of $\frac{|\O(\neg p)|}{|\O(p^\pm)|}$ and let us call this value $q_\0$. By %\textit{monotonicity} and the same reasoning steps as above, we obtain that for any profile %$\O$ in $\AB$: $F(\O)(p)=\0$ iff $\frac{|\O(\neg p)|}{|\O(p^\pm|)}\geq q_\0$. 

\fbox{Claims 2 \& 3} follow straightforwardly from the definitions of uniform quota rule (Definition \ref{def:quota}) and of neutrality (Definition \ref{def:properties}) and, respectively, from the definitions of symmetric quota rules (Definition \ref{def:quota}) and of unbiasedness (Definition \ref{def:properties}) .

\fbox{Claim 4}
Left-to-right. Recall that $\maj$ is defined by quota $\frac{1}{2}< q_\1 =q_\0 \leq \frac{1}{2}+\frac{1}{|N|}$ (Example \ref{example:maj}). It is clear that $\maj$ is uniform and symmetric. The claim then follows by Lemma \ref{lemma:min}. 
Right-to-left. By Lemma \ref{lemma:min} if an aggregator minimizes undecisiveness then its quota are set as $\frac{1}{2}< q_\1 =q_\0 \leq \frac{1}{2}+\frac{1}{|N|}$. These quota define $\maj$ (Example \ref{example:maj}).
\end{proof}

%%%%%%%%%%%%%%%%%%%%%%%%

\begin{proof}[Proof of Proposition \ref{fact:cycles}]
The claim amounts to computing the probability that a random proxy profile $\O$ induces a delegation graph $R_p$ that does not contain gurus (or equivalently, whose endomap $r_p: \N \to \N$ has no fixpoints) as $n$ tends to infinity.
Now, for each agent $i$, the number of possible opinions on a given issue $p$ (that is, functions $O: \set{p} \to \set{\0,\1} \cup \N$) is $|(\N \backslash \set{i}) \cup \set{\0,\1}| = n + 1$ (recall $i$ cannot express ``$i$'' as an opinion). The number of opinions in which $i$ is delegating her vote is $n - 1$. So, the probability that a random opinion of $i$ about $p$ is an opinion delegating $i$'s vote is $\frac{n-1}{n+1}$. Hence the probability that a random profile consists only of delegated votes (no gurus), for a fixed issue, is $(\frac{n-1}{n+1})^n$.
%Notice first of all that, for a given %$|\N|=n$, the probability that an agent %is not a guru is $\frac{n-1}{n}$, and %the probabiity that none of the agents %is a guru $(\frac{n-1}{n})^n$.
%OLD VERSION: the probability that a %delegation graph over $\N$ has no gurus %is \frac{n-1}{n}. 
The claimed value is then established through this series of equations:
\begin{align*}
\lim_{n \to \infty} \left(\frac{n-1}{n+1}\right)^n & = \lim_{n \to \infty} \left(\frac{n}{n+2}\right)^n \\
& = \lim_{n \to \infty} \left(\frac{1}{\frac{n+2}{n}}\right)^n \\
& = \lim_{n \to \infty} \left(\frac{1}{1 + \frac{2}{n}}\right)^n \\
& = \lim_{n \to \infty} \left(\frac{1}{(1 + \frac{2}{n})^n}\right) \\
& = \frac{1}{\lim_{n \to \infty}(1 + \frac{2}{n})^n}\\
& = \frac{1}{e^2}
\end{align*}
This completes the proof.
\end{proof}

\begin{proof}[Proof of Proposition \ref{fact:cycles_default}]
The claim amounts to computing the probability that a random proxy profile with default opinions $\O$ induces a delegation graph $R_p$ (equivalently, an endomap $r_p: \N \to \N$) whose cycles are all hung majorities, that is, whose cycles are all even and exactly half the agents in each cycle accept $p$. As opinion with defaults consist of both a value $x \in \set{\0, \1}$ and a trustee $i \in \N$ we can treat the probability of each component as independent: the number of all possible proxy profiles with default opinions is, therefore, $2^n \cdot n^n$. First of all, recall that a delegation graph can be represented as a set of trees whose roots are nodes in a cycle, that is, as trees whose roots are elements of a permutation of a subset of $\N$. The number of ways of arranging $n$ elements in trees rooted on $m$ elements (with $m>n\geq 1$) is given by the following recursive function (cf. \cite{purdom68cycle}): 
\begin{align}
f(n,m) & = \binom{n}{m} \sum_{0 \leq k \leq n-m} m^k f(n-m,k)
\end{align}
with $f(0,0) = 1$ and $f(n,0) = 0$ for any $n>0$.
So the number $n^n$ of all possible delegation graphs equals
\begin{align}
\sum_{1 \leq k \leq n}f(n,k) k! \label{eq:total}
\end{align}
that is, the number of ways of arranging $n$ elements in trees rooted on a permutation of a subset of $n$ (recall that $k!$ is the number of all possible permutations of $k$ elements). Now to obtain the number of ways of arranging $n$ elements in trees rooted on even cycles, each of which is a hung majority we adapt \eqref{eq:total} as follows. First we establish the number of delegation graphs (for a given issue) which contain only even cycles, that is:
\begin{align}
\sum_{k \leq n \AND \mathit{even}} f(n,k) \frac{k!}{2^k} \binom{k}{\frac{k}{2}}
\end{align}
If each addendum of the above expression is multiplied by $2^k$, that is the number of possible opinions on $p$ of $k$ agents, one obtains the number of possible proxy profiles with default that determine a delegation graph with only even cycles, with all the possible assignments of opinions $x \in \set{\0, \1}$ for the agents in the permutation on which the trees of the graph are rooted:
\begin{align}
\sum_{k \leq n \AND \mathit{even}} f(n,k) k! \binom{k}{\frac{k}{2}} \label{eq:sub}
\end{align}
We can then adapt \eqref{eq:sub} by restricting the subprofiles of opinions of the $k$ agents to hung majorities (i.e., $\binom{k}{\frac{k}{2}}$). We thus obtain the following value:
\begin{align}
\sum_{k \leq n \AND \mathit{even}} f(n,k) \frac{k!}{2^k} \binom{k}{\frac{k}{2}}^2 \label{eq:final}
\end{align}
Under the impartial culture assumption, the probability of a proxy profile with default opinions to induce only even cycles with hung majorities is therefore \eqref{eq:final} divided by $2^n \cdot n^n$.  This quantity approaches $0$ as $n$ tends to infinity.
\end{proof}

%%%%%%%%%%%%%%%%%%%%%%%%

%\begin{proof}[Proof of Fact \ref{fact:influence}]
%\fbox{$2) \Rightarrow 1)$}
%Let $p\in\Atoms$ and assume that $G_p$ contains no cycle of length $\geq 2$ and has diameter $k$. Let $C_p$ be a connected component of $G_p$. By Fact \ref{fact:uniquecycle}, %%%$C_p$ contains a unique cycle, which, by assumption, is of length $1$. Hence, $C_p$ is aperiodic. Let $i$ be the node in the cycle. The opinion of $i$ will spread to all nodes in $C_p%$ %%after at most $k$ steps. Therefore, all BDPs on $G$ will converge after at most $l$ steps, where $l$ is the maximum within the set of diameters of $G_p$ for all $p\in\Atoms$. 
%\fbox{$1) \Rightarrow 3)$} We proceed by contraposition. 
%Assume that for some $p\in\Atoms$, a connected component $C_p$ of $G_p$ contains a cycle of length $k\geq 2$. By \ref{fact:uniquecycle}, this cycle is unique, and therefore the %greatest common divisor of the cycles lengths of $C_p$ is $k$, so $C_p$ is not aperiodic. 
%Let $S$ be the set of nodes in the cycle. 
%Let $\O$ be such that for some $i,j\in S$ with distance $d$ from $i$ to $j$, $O_i(p)\neq O_j(p)$. Then $\O_i(p)$ will not converge, but enter a loop of size $k$: for all $x\in\mathbb{N}$, %$O^{x\times k}_i(p) \neq O^{(x\times k)+d}_i(p)$. Hence, $\O$ does not converge. 
%\fbox{$3) \Rightarrow 2)$} Trivial.
%\end{proof}

\begin{proof}[Proof of Theorem \ref{theorem:resistantsufficient}]
Assume that for all $p\in \Atoms$, for all $S\subseteq\N$ such that $S$ is a cycle in $G_{p}$, for all $i,j\in S$: $O_i(p)=O_j(p)$.
%Consider an arbitrary $p\in \Atoms$, and 
Consider an arbitrary $i\in \N$. 
%Let $k$ be the distance from $i$ to $l$, where $l$ is the closest agent in a cycle $S\subseteq\N$ of $G_p$. 
Let $k_i(p)$ be the distance from $i$ to the closest agent in a cycle of $G_p$, and let $k_i$ denote $max \{k_i(p)|p\in P\}$. We show that for any $k_i\in\mathbb{N}$, $O^{k_i}_ i$ is an opinion which will not change at any later stage (stable).
\begin{itemize}
\item If $k_i=0$: $i$ is its only infuencer, therefore $O^0_{i}$ is stable by assumption. 
\item If $k_i=n+1$: Assume (IH) that for all agents $j$ such that $k_j=n$, $O^{k_j}_ j$ is stable. This implies that all influencers of $i$ are stable. There are two cases:
\begin{enumerate}
\item 
%If $\bigwedge_{p \in \Atoms} O^{m}_{R_p(i)}(p) \wedge \gamma$ is not consistent and closed, 
If $\set{\gamma} \cup \set{p \in \I \mid O_{R_p(i)}(p) = \1} \cup \set{\neg p \in \I \mid O_{R_p(i)}(p) = \0}$
is not consistent, then it will never be, and therefore $O^{n}_i$ is already stable. 
\item 
%If $\bigwedge_{p \in \Atoms} O^{m}_{R_p(i)}(p) \wedge \gamma$ 
If $\set{\gamma} \cup \set{p \in \I \mid O_{R_p(i)}(p) = \1} \cup \set{\neg p \in \I \mid O_{R_p(i)}(p) = \0}$
is consistent, then for each $p$, $O^{n+1}_i(p) = O^{k_j}_ j(p)$, and $O^{n+1}_i$ is therefore (by IH) stable. 
\end{enumerate}
\end{itemize}
It follows that after $k$ steps, with $k= \max\set{diam(G_p)|p\in \I}$, each agent's opinion is stable, and the BDP has therefore stabilized.
 \end{proof}
 
 \begin{proof}[Proof of Theorem \ref{theorem:opinion}]
\fbox{$1) \Rightarrow 2)$} We proceed by contraposition.
Let $p\in\Atoms$, $S\subseteq\N$ be a cycle in $G_p$, $i,j\in S$, and $O_i(p)\neq O_j(p)$. Let $k$ be the length of the cycle and $d$ be the distance from $i$ to $j$. Then $O_i(p)$ will enter a loop of size $k$: for all $x\in\mathbb{N}$, $O^{x k}_i(p)\neq O^{x k+d}_i(p)$. Therefore, the BDP does not stabilize.
\fbox{$2) \Rightarrow 1)$}
Assume $S\subseteq\N$ be such that $S$ is a cycle in $G_p$, and for all $i,j\in S$, $O_i(p)=O_j(p)$. Then, for all $j\in S$, and all $x\in\mathbb{N}$, $O^{x}_j(p)=O_i(p)$ and for all $k \in\N \backslash S$ with distance $d$ from to $i$, for all $x\in\mathbb{N}$, such that $x\geq d$, $O^{x}_k(p)=O_i(p)$. Therefore, the BDP stabilizes.
\end{proof}

%%%%%%%%%%%%%%%%%%%%%%%%%%%%%%%%%%%%

%\nocite{*}
\bibliographystyle{eptcs}
\bibliography{biblio}

\end{document}